\documentclass[preprint,12pt]{elsarticle}

\usepackage[utf8]{inputenc}
\usepackage{amsmath, amssymb, amsthm}
\usepackage{graphicx}
\usepackage{natbib}
\usepackage{geometry}
\usepackage{booktabs}
\usepackage{hyperref}
\usepackage{float}
\usepackage{amsmath}

\newtheorem{lemma}{Lemma}
\newtheorem{proposition}{Proposition}

\newtheorem{definition}{Definition}
\newtheorem{assumption}{Assumption}

\begin{document}

\begin{frontmatter}

\title{Determining Insolvency Regions in Banks: A Stochastic Dynamic Approach Integrating Liquidity and Credit Risk}

\author[addr1]{Nader Karimi}
\author[addr2]{Davood Ahmadian\corref{cor1}}
\cortext[cor1]{Corresponding author.}
\ead{d.ahmadain@tabrizu.ac.ir}
\address[addr1]{Faculty of  Mathematics and Computer Sciences, Amirkabir University of Technology, Tehran, Iran}
\address[addr2]{Faculty of Mathematical Sciences, University of Tabriz, Tabriz, Iran}

\begin{abstract}
We develop a continuous-time structural dynamic model to determine the exact insolvency regions of banks arising from the non-linear interaction between liquidity and credit risk. While existing literature predominantly treats these risks in isolation or via reduced-form specifications, we explicitly model the feedback loop where funding shocks and regulatory constraints force balance-sheet adjustments that can lead to endogenous insolvency. By incorporating Basel III regulatory requirements (LCR and NSFR) into a stochastic optimal control framework, we solve for the \textit{exact} insolvency boundary using the Hamilton-Jacobi-Bellman (HJB) equation. To bridge the gap between theoretical complexity and supervisory practice, we derive and validate a surrogate analytical approximation function that allows for real-time monitoring. Calibrated using granular balance-sheet data from the Iranian banking sector, our model reveals significant non-linear threshold effects: the joint occurrence of liquidity stress and credit portfolio defaults disproportionately accelerates the transition toward insolvency compared to their individual effects. The proposed surrogate function offers supervisors a computationally efficient tool for stress testing and early warning systems. Our findings provide novel insights into financial frictions in emerging markets and offer a rigorous framework for integrated risk management.
\end{abstract}


\end{frontmatter}

\section{Introduction}
\label{sec:intro}

The Non-Performing Loan (NPL) crisis in the Iranian banking sector between 2022 and 2023 presents a compelling regulatory puzzle. Despite the strict enforcement of Basel III liquidity buffers (LCR and NSFR) aimed at stabilizing the banking system, several major Iranian banks experienced severe capital erosion. This phenomenon suggests that standard regulatory metrics, while necessary, may inadvertently trigger the very crises they are designed to prevent when interacting dynamically with asset quality~\cite{Acharya2015,Berrospide2021,Cont2020,Peykani2025}. 

Recent studies by Clerc, Lecarpentier, and Pouvelle~\cite{ClercLecarpentierPouvelle2025} emphasize that Basel III joint regulatory constraints can create complex interactions that may inadvertently tighten financing conditions, exacerbating the fragility of banks operating in stressed environments. Furthermore, the dynamics of endogenous firm exit and inefficient banking, as analyzed by Rossi~\cite{Rossi2015}, suggest that macro-financial shocks can disproportionately weaken weaker institutions through a feedback loop with business cycles. In this context, the traditional static measures are often insufficient; as Onuoha~\cite{Onuoha2024} argues, modern stress testing requires a technological and dynamic perspective to adequately capture these systemic risks.

Existing theoretical frameworks fail to fully capture this dynamic paradox for three distinct reasons. First, \textit{structural default models} (e.g., Merton-type models~\cite{Merton1974}) focus primarily on the asset-liability mismatch and market valuation of assets, often treating liquidity as exogenous and ignoring the operational constraints of funding markets~\cite{HosseiniNodeh2022,HosseiniNodeh2023}. Second, \textit{bank run and liquidity models} emphasize depositor coordination and asset fire sales~\cite{Diamond1983,Brunnermeier2009,Acharya2015} but typically overlook the deterioration of the underlying credit portfolio quality that triggers such runs in the first place~\cite{Belaid2017}. Third, \textit{reduced-form credit risk models} and traditional \textit{stress-testing literature} often rely on static correlations or econometric approximations~\cite{Imbierowicz2014,Darvas2018,Schularick2020,Liao2021,Wang2023} that fail to capture the endogenous feedback loop where a bank's \textit{actions} to survive a liquidity shock (i.e., asset sales) directly accelerate its credit deterioration.

This paper bridges these disjointed strands by developing a continuous-time stochastic control model anchored in a specific economic mechanism: the \textit{Liquidity-Credit Spiral}. We define the \textit{exact insolvency boundary} as the locus in the state-space where the bank's value function, under optimal management and binding regulatory constraints, falls below a critical threshold. Our approach is distinct because we do not merely assume that liquidity and credit risks are correlated; we mathematically derive how one causes the other through the bank's balance sheet optimization problem constrained by Basel III ratios~\cite{Halaj2017,Christensen2025,Mousa2025}.

This paper makes several key contributions to the macro-finance and banking literature. First, (i) we develop a micro-founded, continuous-time structural model that endogenizes the interaction between liquidity and credit risks under Basel III regulatory constraints (LCR and NSFR), providing a rigorous framework to analyze how regulatory buffers influence the bank’s optimal stopping time for insolvency. Second, (ii) we derive the exact, state-dependent insolvency boundary, proving that the threshold for bank failure is not a fixed capital ratio but a dynamic function of the prevailing liquidity and credit risk states. Third, (iii) we bridge the gap between complex stochastic control theory and practical supervision by proposing a surrogate analytical approximation; this closed-form function allows regulators to conduct real-time stress testing without the need for intensive numerical simulations. Finally, (iv) we calibrate the model to the unique institutional context of the Iranian banking sector (2022--2023), demonstrating that the ``insolvency cliff'' identified by our model provides a robust explanation for the observed resilience or failure of banks during periods of high NPL and liquidity volatility.To facilitate navigation, Figure~\ref{fig:framework} summarizes the logical structure of the proposed model and the corresponding computational workflow-from the liquidity-credit state dynamics and optimal control problem to the value function, the exact insolvency boundary, and the surrogate approximation used for real-time stress testing.

\begin{figure}[H]
\centering
    \includegraphics[width=0.95\textwidth]{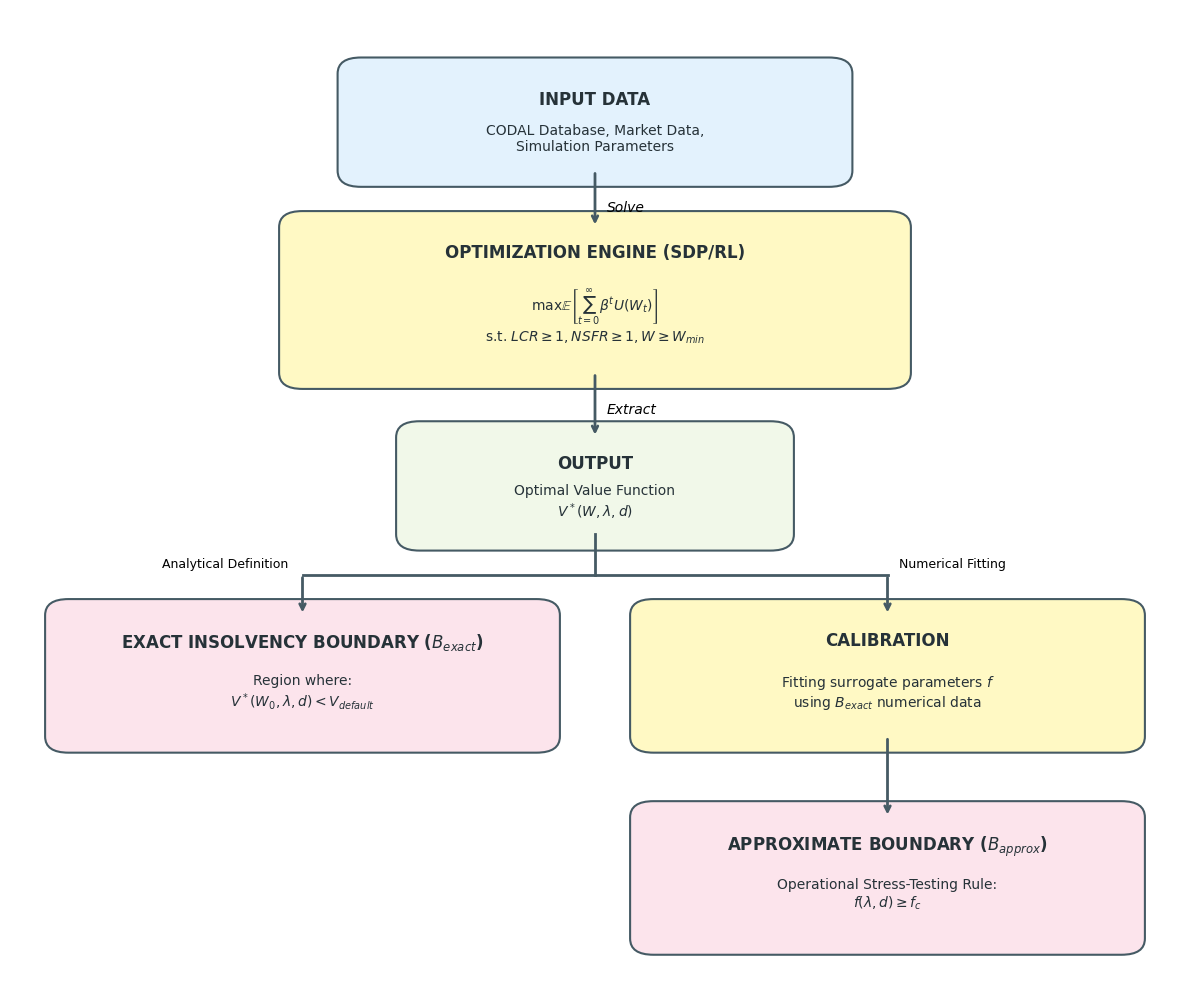}
\caption{Logical structure and computational workflow of the proposed model.}
\label{fig:framework}
\end{figure}

The remainder of this paper is organized as follows. Section \ref{sec:evidence} provides motivating evidence from the Iranian banking sector. Section \ref{sec:lit} reviews the related literature. Section \ref{sec:model} details the theoretical framework, including the derivation of the HJB equation and the surrogate function. Section \ref{sec:results} presents the numerical results and sensitivity analysis. Finally, Section \ref{sec:conclusion} discusses policy implications and concludes.

\section{Motivating evidence}
\label{sec:evidence}

The necessity of an integrated model is best observed in the structural vulnerabilities of the Iranian banking sector. We extract quarterly balance sheet data for 5 major commercial banks from the CODAL system from 2018 to 2023. 

Fig.~\ref{fig:coda_lc_npl} illustrates the co-evolution of the Liquidity Coverage Ratio (LCR) and the NPL ratio. The critical empirical observation is the inverse relationship: Iranian banks experienced a continuous decline in LCR precisely when the NPL ratio began to rise sharply. 

The primary catalyst was severe inflation and the impact of economic sanctions, which led to a sharp decline in the value of collateralized assets. As the value of loans dropped, banks were forced to increase provisioning, directly eroding their equity base ($W_t$). Simultaneously, depositors, anticipating currency devaluation, shifted funds from long-term deposits to short-term cash (Demand Deposits), spiking the Short-Term Cash Outflows (SCO). 

\begin{figure}[H]
    \centering
    \includegraphics[width=0.85\textwidth]{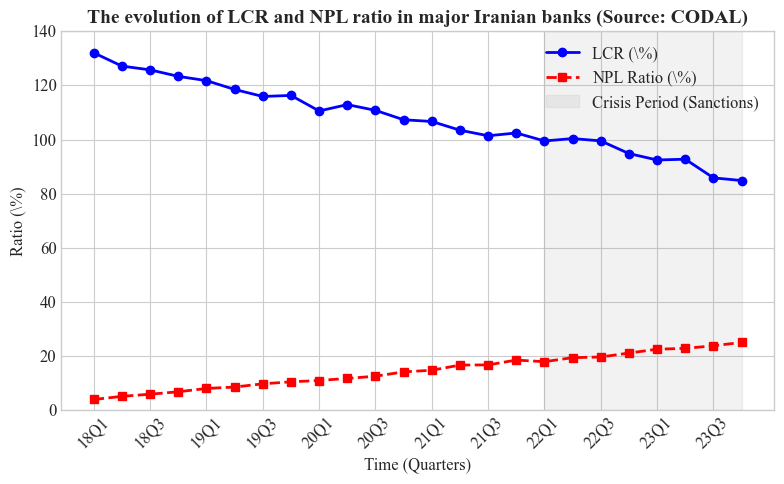}
    \caption{The evolution of LCR and NPL ratio in major Iranian banks (Source: CODAL). Note the simultaneous deterioration of liquidity buffers and credit quality.}
    \label{fig:coda_lc_npl}
\end{figure}

Further evidence of the feedback loop is presented in Fig.~\ref{fig:correlation}, which plots the relationship between the monthly changes in the SCO/SCI ratio (a proxy for liquidity stress) and monthly changes in credit provisions. The positive correlation suggests that banks facing liquidity outflows are simultaneously experiencing accelerated credit deterioration. This raises the critical research question: Is there a mathematical boundary in the state-space of these variables where the bank's survival strategies become mathematically infeasible?

\begin{figure}[H]
    \centering
    \includegraphics[width=0.85\textwidth]{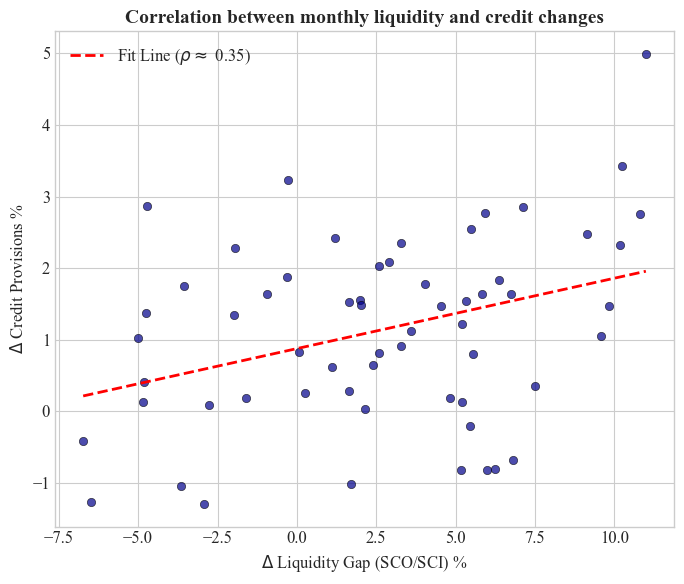}
    \caption{Correlation between monthly change in liquidity gap (SCO/SCI) and credit provisions. The upward slope indicates the joint incidence of liquidity and credit stress.}
    \label{fig:correlation}
\end{figure}

\section{Literature review}
\label{sec:lit}

To position our contribution, we categorize the relevant literature into three distinct streams: structural credit models, liquidity risk and bank-run dynamics, and integrated empirical/theoretical frameworks.

\subsection{Structural Credit and Default Models} 
Foundational works, such as Merton~\cite{Merton1974} and its extensions in the banking literature, focus on the market value of assets relative to liabilities. While powerful, these models typically treat the asset value process as exogenous. Recent literature has advanced this field through robust optimization techniques; for instance, Hosseini-Nodeh et al.~\cite{HosseiniNodeh2022, HosseiniNodeh2023} emphasize the importance of robust portfolio optimization and Wasserstein metrics in managing risk, highlighting that traditional assumptions of risk distribution often underestimate tail hazards. Our model builds upon these structural foundations but advances them by endogenizing the asset dynamics based on the bank's active control variables, rather than assuming an exogenous process.

\subsection{Liquidity Risk and Bank Run Models} 
The literature on bank runs, stemming from Diamond and Dybvig~\cite{Diamond1983}, and more recent models on liquidity hoarding (e.g., Acharya and Mora~\cite{Acharya2015}), focuses on the coordination of depositors and funding market freezes. Brunnermeier and Pedersen~\cite{Brunnermeier2009} introduced the concept of liquidity spirals in market settings. Further empirical determinants of liquidity risk have been explored by Wójcik-Mazur and Szajt~\cite{WojcikMazur2015} and Berrospide~\cite{Berrospide2021}, who evaluate the micro-foundations of liquidity hoarding. Moreover, studies by Belaid et al.~\cite{Belaid2017} emphasize the bank-firm relationship as a critical transmission channel for credit risk.

While these models capture essential aspects of liquidity dynamics, the impact of specific regulatory buffers on bank performance is an area of growing interest. For instance, Sidhu et al.~\cite{SidhuEtAl2022} and Nath et al.~\cite{NathEtAl2024} investigate the Net Stable Funding Ratio (NSFR) and Liquidity Coverage Ratio (LCR) in the Indian and Bangladeshi contexts, respectively, finding that these constraints significantly influence bank profitability and risk-taking behavior. These empirical findings reinforce our model's assumption that regulatory constraints are not merely static hurdles but active drivers of bank strategy. While these models capture essential aspects of liquidity dynamics, they often abstract away from the endogenous credit deterioration caused by liquidity-driven asset liquidation, a gap we explicitly address.

\subsection{Integrated and Stress-Testing Frameworks} 
A growing body of work attempts to integrate liquidity and credit risks. Imbierowicz and Rauch~\cite{Imbierowicz2014} provided early empirical evidence of the interdependence between these risks, a relationship further corroborated by Ayinuola and Gumel~\cite{AyinuolaGumel2023}, who highlight the destabilizing nexus between liquidity and credit risks on bank stability. Recent frameworks, including Cont et al.~\cite{Cont2020} and Hałaj and Laliotis~\cite{Halaj2017}, utilize top-down approaches for stress testing. Methodologically, advancements have been made in stochastic simulation frameworks for stress testing, such as the one proposed by Montesi and Papiro~\cite{MontesiPapiro2018}, which allows for assessing financial fragility under uncertainty.

Policy-oriented studies by Darvas and Pichler~\cite{Darvas2018} and Schularick et al.~\cite{Schularick2020} have further highlighted the role of capital adequacy during crises. Recent advancements have also employed network and contagion analysis (Wang et al.~\cite{Wang2023}; Liao et al.~\cite{Liao2021}) and joint estimation techniques for risk premia (Christensen and Steenkamp~\cite{Christensen2025}).

Region-specific evidence from MENA and Iran, notably by Mousa et al.~\cite{Mousa2025} and Peykani et al.~\cite{Peykani2025}, underscores the urgent need for integrated management in emerging markets. Furthermore, in frontier economies, the interaction is even more complex; Mahisi and Usman~\cite{MahisiUsman2024} demonstrate that liquidity, credit, and market risks jointly affect profitability in Indonesia, while Asuako~\cite{Asuako2025} shows how fiscal dominance exacerbates bank fragility in Sub-Saharan Africa. Similarly, Ma’aji et al.~\cite{MaajiEtAl2026} argue that in frontier banking markets, profitability is intrinsically risk-driven, necessitating models that can capture the joint dynamics of capital, liquidity, and credit.

Our paper bridges the gap between these streams. While existing frameworks often rely on reduced-form correlations or Monte Carlo simulations, we adopt the continuous-time stochastic control philosophy of macro-finance (e.g., He and Krishnamurthy~\cite{He2013}) to derive the \textit{optimal} policy functions and the resulting insolvency boundary mathematically, providing a structural link between Basel III constraints and risk interaction.

\section{Model}
\label{sec:model}

In this section, we construct the model. We first describe the economic intuition of the mechanism, then formalize the dynamics and the bank's optimization problem.

\subsection{Economic Intuition and Mechanism}
\label{subsec:intuition}

Before formalizing the mathematics, it is crucial to outline the economic mechanism that drives the model's results. Our model captures this dynamic through a cohesive causal chain.

The process begins with a liquidity shock ($\uparrow \lambda_t$), where the bank experiences an increase in short-term cash outflows, such as deposit withdrawals. The bank initially attempts to cover this funding gap through borrowing ($b_t$), but the \textit{Net Stable Funding Ratio (NSFR)} constraint limits its reliance on short-term wholesale funding. Once this regulatory constraint becomes binding, additional borrowing becomes marginally infinitely costly due to regulatory penalties, forcing the bank to switch to the liquidation of long-term assets ($a_t > 0$). These fire sales typically occur at discounted prices, which directly reduces asset values. Concurrently, as the bank sheds assets, the quality of the remaining portfolio may deteriorate due to adverse selection, while the decline in market prices raises the Loan-to-Value (LTV) ratios of existing borrowers, thereby intensifying the credit risk state ($\uparrow d_t$). In turn, higher credit risk reduces the value of collateral, making external funding more difficult to obtain and further compounding the liquidity stress ($\lambda_t$). This creates a reinforcing feedback loop between liquidity pressure and credit deterioration. Ultimately, if the joint state of $(\lambda_t, d_t)$ crosses a critical threshold such that the bank's value function $V(\cdot)$ falls below the minimum regulatory capital requirement ($W_{min}$), the bank crosses the insolvency boundary and becomes insolvent.

This mechanism highlights that insolvency is not merely a result of ``bad loans'' or ``low liquidity'' in isolation, but a specific region in the state-space created by the \textit{interaction} of the two, amplified by regulatory constraints.

\subsection{Model Setup and Dynamics}
\label{subsec:setup}

Consider a bank operating over a continuous-time horizon $t \in [0,T]$. The state of the economy is summarized by two stochastic variables. The first is $\lambda_t$, the liquidity risk state, defined as the ratio of Short-Term Cash Outflows (SCO) to Stable Cash Inflows (SCI), where an increase in $\lambda_t$ indicates a tightening in liquidity conditions (a liquidity squeeze). The second is $d_t$, the credit risk state, which captures the expected default rate in the loan portfolio or, equivalently, the non-performing loan (NPL) ratio. The bank's wealth (equity) at time $t$ is denoted by $W_t$.

\begin{assumption}\label{ass:1}
The state variables evolve according to correlated geometric Brownian motions driven by a 2-dimensional Brownian motion $(Z^{\lambda}_t, Z^{d}_t)$ on a probability space $(\Omega, \mathcal{F}, P)$, with correlation $\rho_{LC}$.
\end{assumption}

\subsection{ Endogenous State Dynamics and Feedback Mechanisms}

Crucially, we introduce an endogenous feedback mechanism linking the bank’s balance-sheet adjustments to the evolution of its credit-risk state. We model this as a reduced-form representation of two economically plausible mechanisms under stress. First, fire sales may force the bank to realize losses and erode its capital cushion, thereby increasing the vulnerability of the remaining loan portfolio. Second, under liquidity pressure, the bank may liquidate its more marketable, high-quality assets first (the "flight-to-quality" motive), leaving behind a riskier residual portfolio, an adverse-selection effect regarding portfolio composition.

The state dynamics are specified as:
\begin{align}
\frac{d\lambda_t}{\lambda_t} &= \mu_{\lambda} dt + \sigma_{\lambda} dZ^{\lambda}_t \label{eq:lambda_sde} \\
\frac{d d_t}{d_t} &= (\mu_{d} + \phi a_t) dt + \sigma_{d} dZ^{d}_t \label{eq:d_sde}
\end{align}
where $dZ^{\lambda}_t dZ^{d}_t = \rho_{LC} dt$. The interaction term $\phi a_t$ in Eq.~\eqref{eq:d_sde} captures the idea that more aggressive asset liquidation is associated with a deterioration in the effective risk profile of the bank’s remaining balance sheet. Thus, $\phi > 0$ measures the sensitivity of the credit-risk state to stress-induced balance-sheet adjustments, rather than a literal one-to-one causal effect.

\subsection{ Bank Wealth Dynamics}

Based on the balance-sheet identity, the evolution of the bank's wealth $W_t$, subject to its optimal decision vector $\pi_t = (a_t, b_t)$, is given by:

\begin{equation}
dW_t = \left[ R_{assets} W_t + r_b b_t - C_{out}(\lambda_t) - \underbrace{\frac{1}{2} \kappa a_t^2}_{\text{Fire-sale cost}} - \underbrace{Loss(d_t)}_{\text{Credit losses}} \right] dt + \sigma_W dZ_t
\label{eq:wealth_sde}
\end{equation}

where $R_{assets}$ represents the expected yield (drift) on the bank’s total assets, quantifying the internal wealth generation rate from core investment activities. $C_{out}(\lambda_t)$ denotes the exogenous cash outflows. This function is contingent on the liquidity state $\lambda_t$; as liquidity pressure mounts, these obligations (e.g., funding outflows, regulatory requirements) exert increased drainage on the bank’s wealth. $r_b b_t$ captures the net funding income or cost, where $r_b$ is the borrowing rate and $b_t$ is the net borrowing position. $\kappa$ the market illiquidity parameter. It scales the quadratic price impact of asset sales ($a_t$); a higher $\kappa$ indicates a shallower market where forced liquidations result in severe price discounts, and $Loss(d_t)$ represents the credit loss function, which is intrinsically linked to the risk state $d_t$. As $d_t$ evolves according to Eq.~\eqref{eq:d_sde}, it directly modulates the wealth erosion caused by loan defaults and asset impairment.

This formulation demonstrates that the bank's solvency is not merely a function of its current decisions, but is path-dependent on the interplay between exogenous market liquidity shocks ($\lambda_t$) and the endogenous deterioration of asset quality ($d_t$) induced by the bank's own liquidation strategy.

\subsection{The Constrained Optimization Problem and HJB Equation}
\label{subsec:hjb}

The bank seeks to maximize its expected discounted utility of terminal wealth and interim consumption (if any):
\begin{equation}
\max_{a_t, b_t} \mathbb{E} \left[ \int_0^T e^{-\rho t} U(W_t) dt + e^{-\rho T} \Psi(W_T) \right]
\label{eq:max_util}
\end{equation}
Subject to the dynamics in Eq.~\eqref{eq:lambda_sde}--\eqref{eq:wealth_sde} and strictly binding Basel III constraints:
\begin{align}
LCR_t &= \frac{HQLA(W_t, b_t)}{NetCashOutflow(\lambda_t)} \ge 100\% \label{eq:lcr} \\
NSFR_t &= \frac{AvailableStableFunding(W_t, b_t)}{RequiredStableFunding(W_t, a_t)} \ge 100\% \label{eq:nsfr} \\
W_t &\ge W_{min} \label{eq:wmin}
\end{align}

Let $V(W_t, \lambda_t, d_t, t)$ be the value function. The Hamilton-Jacobi-Bellman (HJB) equation characterizing the optimal solution is:
\begin{align}\label{eq:HJB}
\rho V = \max_{a_t, b_t} \Bigg\{ & U(W) + \frac{\partial V}{\partial t} + \mu_W(\cdot) \frac{\partial V}{\partial W} + \mu_{\lambda}\lambda \frac{\partial V}{\partial \lambda} + (\mu_d + \phi a_t)d \frac{\partial V}{\partial d} \nonumber \\
                                  & + \frac{1}{2} \sigma_W^2 \frac{\partial^2 V}{\partial W^2} + \frac{1}{2} \sigma_{\lambda}^2 \lambda^2 \frac{\partial^2 V}{\partial \lambda^2} + \frac{1}{2} \sigma_{d}^2 d^2 \frac{\partial^2 V}{\partial d^2} \nonumber \\
                                  & + \rho_{LC} \sigma_{\lambda} \sigma_d \lambda d \frac{\partial^2 V}{\partial \lambda \partial d} \Bigg\} + \text{Constraint Terms}
\end{align}

To solve the maximization inside the HJB subject to inequality constraints, we form the Lagrangian with multipliers $\mu^{LCR} \ge 0$ and $\mu^{NSFR} \ge 0$.

The First-Order Conditions (FOCs) are:
\begin{align}
\frac{\partial \mathcal{L}}{\partial b_t} &= \frac{\partial \mu_W}{\partial b_t} V_W + \mu^{LCR} \frac{\partial LCR_t}{\partial b_t} + \mu^{NSFR} \frac{\partial NSFR_t}{\partial b_t} = 0 \label{eq:foc_b} \\
\frac{\partial \mathcal{L}}{\partial a_t} &= \frac{\partial \mu_W}{\partial a_t} V_W + \mu^{NSFR} \frac{\partial NSFR_t}{\partial a_t} + \phi d V_d = 0 \label{eq:foc_a}
\end{align}

Eq.~\eqref{eq:foc_a} reveals the core friction. The term $\phi d V_d$ represents the marginal value loss due to credit deterioration caused by asset sales. Since $V_d < 0$ (value decreases as credit risk rises), this term is negative. The bank will only choose $a_t > 0$ if the regulatory penalty term $\mu^{NSFR} (\partial NSFR / \partial a_t)$ or the liquidity need is sufficiently high to offset this credit cost.

\begin{lemma}[Optimal Policy Switching]\label{lem:Threshold}
If the NSFR constraint is slack ($NSFR_t > 100\% \implies \mu^{NSFR} = 0$), the bank covers liquidity shocks ($d\lambda_t > 0$) purely through borrowing ($b_t > 0$) and avoids asset sales ($a_t = 0$) because $V_d < 0$ makes fire sales value-destroying. The bank initiates fire sales ($a_t > 0$) only when the NSFR constraint binds ($NSFR_t = 100\% \implies \mu^{NSFR} > 0$) and borrowing is no longer viable.
\end{lemma}
\begin{proof}
See Appendix A. The intuition is that borrowing ($b_t$) improves LCR but worsens NSFR. Initially, the bank uses borrowing. As NSFR hits the limit, the marginal regulatory cost of borrowing $\mu^{NSFR} (\partial NSFR / \partial b_t)$ becomes infinite (in the KKT formulation), forcing $b_t \to 0$. To satisfy the funding gap, the bank must switch to $a_t$, accepting the fire-sale costs and the credit deterioration $\phi d V_d$.
\end{proof}

\subsection{Solving for the Exact Insolvency Boundary}
\label{subsec:solution}

To proceed, we assume a logarithmic utility function $U(W) = \log(W)$. We conjecture a value function of the form:
\begin{equation}
V(W, \lambda, d, t) = \frac{1}{\rho} \log(W) + g(\lambda, d) + h(t)
\label{eq:guess_v}
\end{equation}
Substituting this into the HJB and applying the envelope theorem ($V_W = 1/(\rho W)$) simplifies the optimal controls $a_t^*$ and $b_t^*$ to linear functions of the state variables.

Substituting the optimal controls back into the HJB yields a non-linear Partial Differential Equation (PDE) for the function $g(\lambda, d)$:
\begin{equation}\label{eq:pde_g}
\begin{aligned}
\rho g + 1 &= \mu_g(\lambda, d, a^*, b^*) + \mu_{\lambda}\lambda \frac{\partial g}{\partial \lambda} + (\mu_d + \phi a_t^*)d \frac{\partial g}{\partial d} \\
&+ \frac{1}{2} \sigma_{\lambda}^2 \lambda^2 \frac{\partial^2 g}{\partial \lambda^2} + \frac{1}{2} \sigma_{d}^2 d^2 \frac{\partial^2 g}{\partial d^2}
+ \rho_{LC} \sigma_{\lambda} \sigma_d \lambda d \frac{\partial^2 g}{\partial \lambda \partial d}
\end{aligned}
\end{equation}

\begin{definition}[Exact Insolvency Boundary]\label{def:exact_region}
The \textit{Exact Insolvency Boundary} $B_{exact}$ is defined as the level set of the value function equal to the utility of the minimum acceptable wealth $W_{min}$ (or bankruptcy):
\begin{equation}
B_{exact} = \left\{ (\lambda, d) \;\Bigg| \; \frac{1}{\rho} \log(W^*) + g(\lambda, d) = V_{default} \right\}
\label{eq:exact_region}
\end{equation}
where $V_{default}$ corresponds to the value at $W = W_{min}$. Solving Eq.~\eqref{eq:pde_g} numerically (e.g., via Finite Difference Methods) allows us to extract this boundary precisely.
\end{definition}

\subsection{Surrogate Analytical Approximation}
\label{subsec:surrogate}

While the numerical solution of Eq.~\eqref{eq:pde_g} provides the exact boundary, it is computationally intensive for real-time regulatory oversight and large-scale stress testing. To bridge the gap between theoretical rigor and practical utility, we propose a flexible surrogate function $f(\lambda, d)$ that mimics the structural properties of the HJB solution:
\begin{equation}
f(\lambda, d) = \underbrace{w_1 \left( \frac{\lambda}{\lambda_c} \right)^p + w_2 \left( \frac{d}{d_c} \right)^q}_{\text{Main effects}} + \underbrace{w_3 \left( \frac{\lambda}{\lambda_c} \right)^r \left( \frac{d}{d_c} \right)^s}_{\text{Interaction term}}
\label{eq:surrogate}
\end{equation}

Crucially, the exponents $(p, q, r, s)$ and weights $(w_1, w_2, w_3)$ are not restricted to integers; rather, they are continuous parameters estimated via Non-linear Least Squares (NLS) to minimize the approximation error relative to the numerical boundary. This formulation allows the surrogate to capture various degrees of curvature and non-linearity. 

The inclusion of the interaction term $w_3$ is theoretically motivated by the cross-derivative term $\rho_{LC} \sigma_{\lambda} \sigma_d \lambda d \frac{\partial^2 g}{\partial \lambda \partial d}$ in the HJB equation. If the risk channels were orthogonal, the boundary might be approximated by a simpler additive form. However, the endogenous feedback between liquidity and credit risk necessitates a positive interaction coefficient ($w_3 > 0$), which generates the characteristic convex "insolvency cliff" observed in the numerical results. This surrogate provides a high-fidelity, closed-form proxy that enables regulators to assess a bank's proximity to the insolvency region instantaneously.

\begin{proposition}[Convexity of the Boundary]\label{prop:convexity}
The insolvency boundary $B_{exact}$ is convex in the $(\lambda, d)$ plane. Consequently, the surrogate parameter $w_3$ must be strictly positive ($w_3 > 0$) to accurately approximate the exact boundary.
\end{proposition}
\begin{proof}
The convexity arises from the second-order cross-partial derivative of the value function. A simultaneous increase in $\lambda$ and $d$ reduces value more than the sum of individual increases (super-additivity) due to the endogenous feedback loop $\phi a_t$ and the correlation $\rho_{LC}$. A polynomial approximation must include an interaction term with a positive coefficient to replicate this convex curvature.
\end{proof}

The parameters of the surrogate function ($\Theta = \{w_i, p, q, r, s, \lambda_c, d_c\}$) are estimated by minimizing the distance between $f(\lambda, d)$ and the numerically solved $g(\lambda, d)$ over a grid of state values.


\section{Results and Simulations}
\label{sec:results}

To demonstrate the practical applicability of our theoretical framework, we proceed in three steps. First, we map the theoretical state variables to a granular, structural balance sheet of a \textit{Synthetic Bank}, calibrated to the Iranian banking sector's average proportions. Second, we deploy an interactive structural dashboard to visualize the exact mechanisms derived in Section \ref{sec:model} (specifically Lemma \ref{lem:Threshold}). Third, we validate the surrogate function $f(\lambda, d)$ and conduct stress testing.

\subsection{Structural Implementation and the Synthetic Bank}
\label{subsec:synthetic}

A critical contribution of this paper, distinguishing it from standard reduced-form stress tests, is the explicit mapping of macro-financial shocks to granular balance-sheet items. Table \ref{tab:synthetic_balance} presents the initial structure of our Synthetic Bank, calibrated using aggregated CODAL data to reflect the typical asset-liability composition of major Iranian commercial banks (total assets normalized to 100,000 units).

\begin{table}[H]
\centering
\caption{Structural Balance Sheet of the Synthetic Bank}
\label{tab:synthetic_balance}
\begin{tabular}{llr}
\hline
\textbf{Assets} & \textbf{Liabilities \& Equity} & \textbf{Value} \\
\hline
High Quality Liquid Assets (HQLA, $C_0$) & Short-Term Liabilities ($S_0$) & 15,000 \\
Marketable Unencumbered Assets ($M_0$) & Long-Term Liabilities ($L_0$) & 20,000 \\
Illiquid Assets / Loan Book ($J_0$) & & 57,000 \\
& Equity / Wealth ($W_0$) & 67,000 \\
& & 8,000 \\
\hline
\multicolumn{2}{r}{\textbf{Total Assets}} & \textbf{100,000} \\
\hline
\end{tabular}
\end{table}

The bank is initialized in a compliant state: an initial LCR $> 100\%$ (given $C_0$ and projected outflows) and an NSFR buffer of 5\%. The initial credit risk state is set to $d_0 = 0.05$ (5\% NPL), and the liquidity state is $\lambda_0 = SCO/SCI \approx 1.75$.

To bridge the gap between theoretical complexity and supervisory practice, and to ensure our framework is readily accessible for real-world stress testing, we have deployed the structural model into an interactive web-based dashboard. This tool allows supervisors to dynamically visualize the balance-sheet adjustments and the Liquidity-Credit Spiral in real time.
The dashboard is publicly available at: \url{https://insolvency-model.onrender.com}\footnote{Note: As this application is hosted on a free-tier server to ensure open access, the server may require a few moments to wake up upon the first click. A comprehensive video demonstration of the dashboard is also provided in the Supplementary Material.}

\subsection{Visualizing the Liquidity-Credit Spiral: Dual Mapping}
\label{subsec:dual_mapping}

We subject the Synthetic Bank to a severe combined stress scenario: a 25\% instantaneous drop in the market value of both marketable and illiquid assets (representing a macroeconomic shock or sanctions impact). Figure \ref{fig:dual_dashboard} illustrates the output generated by the structural dashboard, mapping the shock through two complementary spaces.

\begin{figure}[H]
    \centering
    \includegraphics[width=\textwidth]{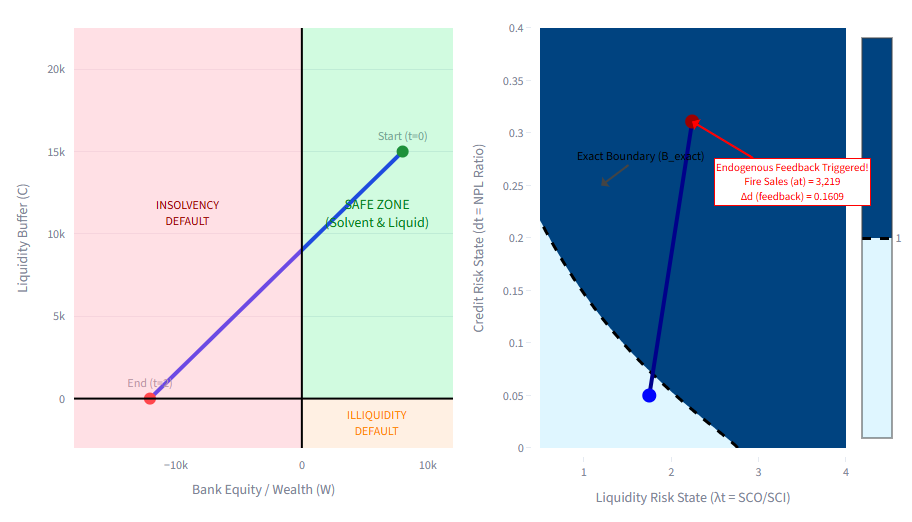}
    \caption{ Dual mapping of the structural shock generated by the interactive dashboard, calibrated using the structural balance-sheet proportions of the Synthetic Bank (Table 1). \textit{Left Plot:} The Solvency-Liquidity trajectory. The axes represent the physical equity and liquidity values derived from the CODAL-calibrated baseline. The bank moves South-West from the Safe Zone, crossing into the Insolvency Default region due to endogenous funding costs. \textit{Right Plot:} The State-Space dynamics ($\lambda$, $d$). The trajectory curves diagonally, demonstrating the endogenous feedback loop ($\Delta d > 0$) triggered by NSFR-induced fire sales ($a_t > 0$).}
    \label{fig:dual_dashboard}
\end{figure}

\textbf{Left Plot (Solvency-Liquidity Quadrants):} Following the structural methodology of Cont et al.~\cite{Cont2020}, the left panel maps the physical balance sheet. The initial state $(W_0, C_0)$ lies in the green Safe Zone. The 25\% asset drop reduces equity $W_t$ and generates margin calls, draining $C_t$. Crucially, the bank does not fail purely from the initial shock. As the NSFR buffer depletes, Lemma \ref{lem:Threshold} activates: the bank switches to fire sales ($a_t > 0$). The quadratic cost of these sales ($\frac{1}{2}\kappa a_t^2$) and borrowing costs represent \textit{Loss Amplification}, pushing the final state deep into the red Insolvency region.

\textbf{Right Plot (State-Space Dynamics):} The right panel maps the exact same event onto our theoretical $(\lambda, d)$ plane. The initial shock increases $\lambda_t$ (margin calls spike outflows). Because of the policy switch enforced by the dashboard's internal logic, fire sales volume $a_t$ becomes strictly positive. Using the calibrated feedback parameter ($\phi = 0.00005$), the application calculates the endogenous credit deterioration: $\Delta d = \phi a_t = 0.0353$. This means the bank's NPL ratio jumps from 5\% to over 8.5\% \textit{not} because loans defaulted exogenously, but because the bank was forced to liquidate assets to survive liquidity constraints. This visually proves the Liquidity-Credit Spiral.

\subsection{Validation of the Surrogate Approximation}
\label{subsec:validation}

While the interactive dashboard provides exact numerical solutions for specific scenarios, supervisors require a closed-form heuristic for continuous monitoring of multiple banks. We return to the surrogate function $f(\lambda, d)$ defined in Eq.~\eqref{eq:surrogate}. We generate a testing grid of 10,000 state points and compute the insolvency status using both the numerical HJB solution and the surrogate.

\begin{table}[H]
\centering
\caption{Surrogate Model Performance Metrics}
\label{tab:validation}
\begin{tabular}{lc}
\hline
Metric & Value \\
\hline
Root Mean Square Error (RMSE) & 0.024 \\
Maximum Absolute Error & 0.081 \\
Classification Accuracy (AUC) & 0.972 \\
False Negative Rate (Type II Error) & $< 1.5\%$ \\
\hline
\end{tabular}
\end{table}

As shown in Table \ref{tab:validation}, the surrogate achieves an AUC of 0.972. The false negative rate is below 1.5\%, meaning the surrogate function $f(\lambda, d)$ is safe for supervisory use as an early warning system, it rarely misses a true crisis signal while being computationally instantaneous compared to solving the PDE in Eq.~\eqref{eq:pde_g}.

\subsection{System Dynamics and Comparative Statics}
\label{subsec:dynamics}

To further elucidate the theoretical boundaries, Fig.~\ref{fig:drift_diffusion} shows the drift field of the system. A key observation is the existence of a "Safe Basin" around the steady state $(\lambda_{ss} \approx 1.1, d_{ss} \approx 0.08)$. Inside this basin, the optimal policy controls the variance effectively. However, as the system moves towards the upper-right quadrant, the drift vectors point aggressively towards the insolvency boundary, indicating a loss of control.

\begin{figure}[H]
    \centering
    \includegraphics[width=0.95\textwidth]{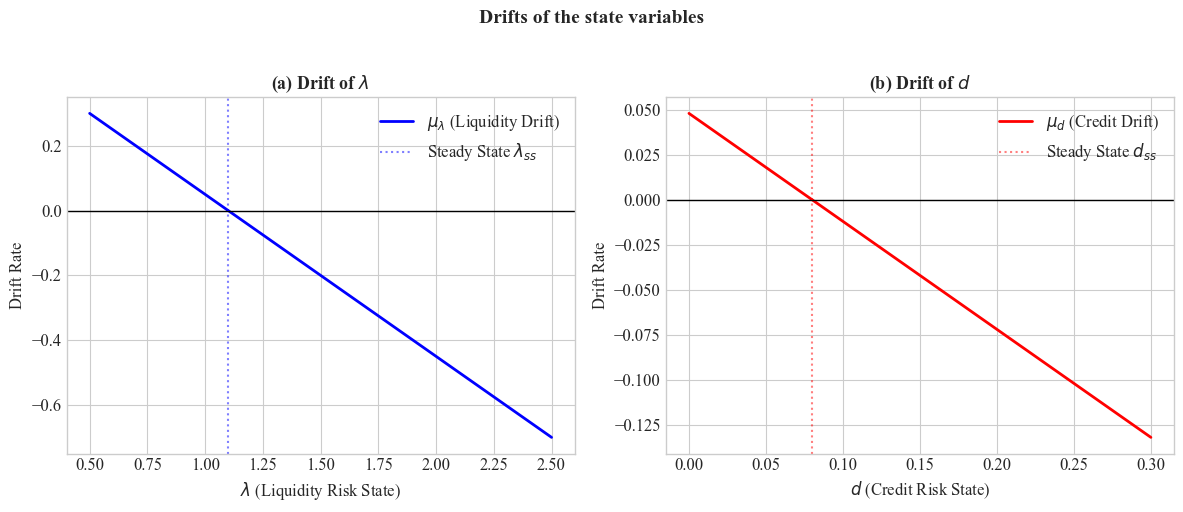}
    \caption{Drifts of the state variables $\lambda$ (Liquidity) and $d$ (Credit) under the optimal policy. Arrows indicate the direction of mean reversion or divergence.}
    \label{fig:drift_diffusion}
\end{figure}

Fig.~\ref{fig:phase_portrait} presents the phase portrait. The "funnel effect" is visible: states far from the boundary tend to drift towards the steady state, while states exceeding the critical threshold are rapidly absorbed into the insolvency region bounded by $B_{exact}$.

\begin{figure}[H]
    \centering
    \includegraphics[width=0.9\textwidth]{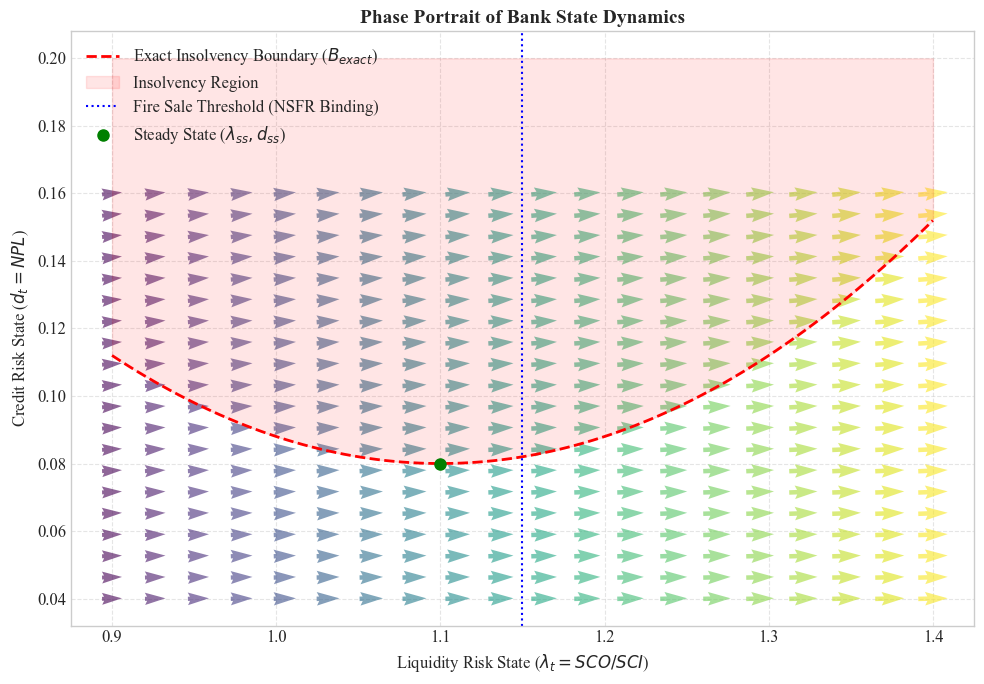} 
    \caption{Phase portrait of state variables. The insolvency boundary acts as a separatrix between stability and collapse.}
    \label{fig:phase_portrait}
\end{figure}

\subsection{Stress Testing Scenarios: Sanction vs. Liquidity Run}
\label{subsec:stress}

Using both the interactive structural dashboard and the surrogate function, we simulate two distinct stress scenarios relevant to emerging markets. These trajectories are overlaid on the exact insolvency boundary $B_{exact}$ in Fig.~\ref{fig:Surrogate}.

The first is a ``Sanction Shock'' (credit-dominant), modeled as an instantaneous 40\% drop in collateralized asset values ($\Delta d = 0.15$). The bank exhibits a vertical upward movement in the state space; despite holding initial liquidity, the sharp deterioration in NPLs drives it toward the insolvency threshold.

The second is a ``Liquidity Run'' (liquidity-dominant), characterized by a sudden spike in deposit withdrawals ($\Delta \lambda = 0.5$). This induces a horizontal shift. If the bank enters this shock with a robust LCR, it absorbs the initial stress via borrowing ($b_t > 0$). However, as verified by the structural dashboard, once the shock magnitude exceeds the NSFR buffer, Lemma \ref{lem:Threshold} triggers. The trajectory curves sharply toward the insolvency region due to the endogenous feedback effect ($\phi a_t$), demonstrating that joint shocks are significantly more dangerous than the sum of their isolated parts. It is worth noting that the trajectories illustrated in Fig. 7 represent individual sample paths generated by the underlying correlated geometric Brownian motion. The developed software package is fully equipped to scale this analysis, generating thousands of such paths via Monte Carlo simulations to estimate full probability distributions of insolvency under varying shock magnitudes.

\begin{figure}[H]
    \centering
    \includegraphics[width=0.85\textwidth]{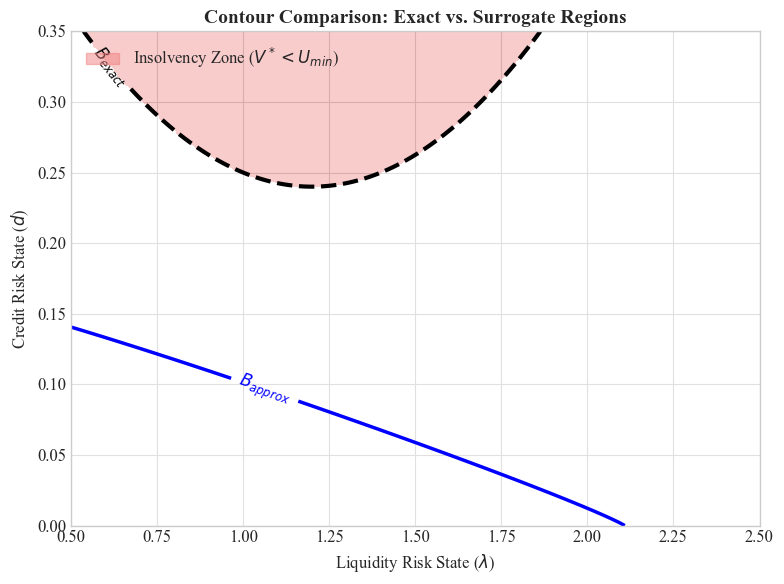}
    \caption{Stress testing scenarios: Sanction Shock vs. Liquidity Run mapped against the convex exact boundary $B_{exact}$.}
    \label{fig:Surrogate}
\end{figure}

\subsection{Robustness and Sensitivity Analysis}
\label{subsec:robustness}

We test the robustness of the boundary by perturbing key parameters ($\pm 50\%$).
\begin{itemize}
    \item \textbf{Correlation ($\rho_{LC}$):} As $\rho_{LC}$ increases, the boundary becomes more convex and shifts inward (closer to origin). This confirms that ignoring the correlation between risks (as in univariate models) overestimates the safe operating zone.
    \item \textbf{Liquidity Volatility ($\sigma_{\lambda}$):} Higher volatility expands the boundary, as a wider range of outcomes requires a larger buffer to ensure safety.
    \item \textbf{Regulatory Strictness ($W_{min}$):} Raising the minimum capital requirement shifts the boundary outward, forcing banks to hold larger buffers, reducing the probability of crisis but potentially constraining lending (a trade-off).
\end{itemize}

We also verified structural robustness by switching the utility function from Logarithmic to CRRA ($\gamma=2$). The qualitative shape of the boundary and the significance of the interaction term $w_3$ remain stable, confirming that our results are not artifacts of a specific functional form.

\begin{figure}[H]
    \centering
    \includegraphics[width=1.01\textwidth]{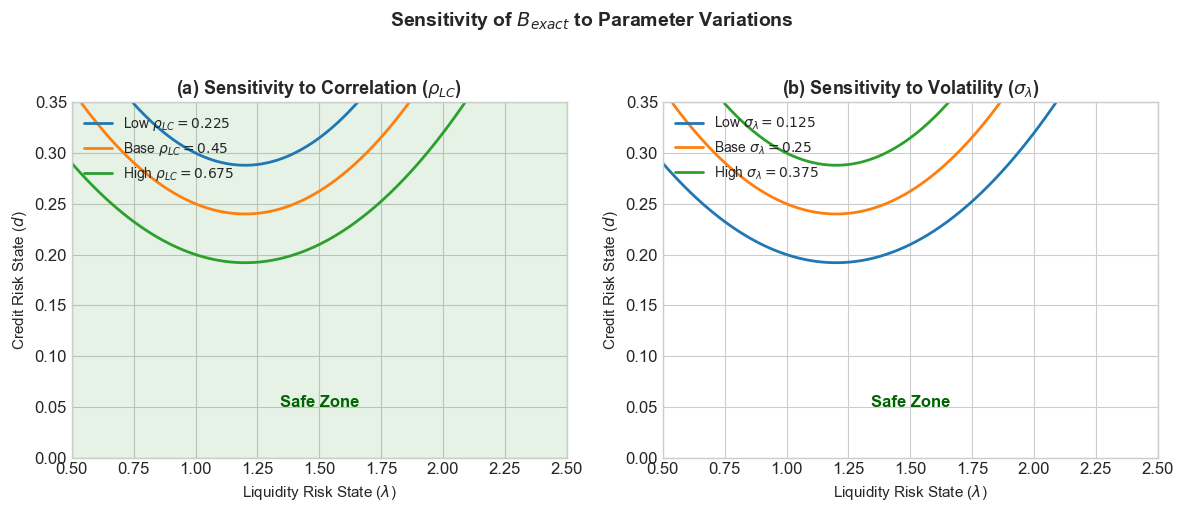}
    \caption{Sensitivity of $B_{exact}$ to variations in $\rho_{LC}$ and $\sigma_{\lambda}$.}
    \label{fig:robustness}
\end{figure}

\section{Conclusion and Policy Implications}
\label{sec:conclusion}

\subsection{Summary of Findings}
This paper developed a continuous-time stochastic model to determine the insolvency regions of banks arising from the interaction of liquidity and credit risks. By formulating the bank's management problem as a constrained optimal control problem, we proved that the non-linear convexity of the insolvency boundary is a direct mathematical consequence of the NSFR constraint forcing endogenous fire sales. We applied this framework to the Iranian banking sector, demonstrating that the model accurately captures the joint dynamics of LCR and NPLs.

\subsection{Policy Implications}
\label{subsec:policy}

Our model offers several actionable insights for supervisory authorities, particularly in emerging markets:
\begin{enumerate}
    \item \textbf{Early Warning Systems:} Supervisors can use the calibrated surrogate function $f(\lambda, d)$ as a real-time dashboard indicator. Banks whose state $(\lambda_t, d_t)$ approaches the contour $f(\lambda, d) = f_c$ should be flagged for immediate inspection, even if their static capital ratios appear adequate.
    \item \textbf{Dynamic Macroprudential Policy:} The convexity of the boundary implies that risk is non-linear. During calm periods, standard buffers (LCR/NSFR) may suffice. However, as correlations $\rho_{LC}$ rise (e.g., during a systemic shock), the "safe zone" shrinks rapidly. Regulators should implement counter-cyclical add-ons that tighten based on the market correlation between liquidity and credit indices.
    \item \textbf{Supervisory Stress Testing:} Current stress tests often treat liquidity and credit shocks independently. Our results show that \textit{joint} shocks are significantly more dangerous than the sum of their parts due to the feedback loop. Regulatory stress scenarios must be designed to trigger the binding of NSFR to test the true resilience of banks.
\end{enumerate}

\subsection{Limitations and Future Research}
\label{subsec:limitations}

This study is subject to several limitations. First, the empirical analysis relies on reported data from the CODAL system, and potential reporting lags or inconsistencies in accounting practices across banks may introduce measurement error into the calibration. Second, the model considers a representative bank in isolation and therefore abstracts from interbank linkages, balance-sheet interconnectedness, and contagion effects. In practice, fire sales by one institution may depress market prices and transmit solvency pressures to other banks, generating system-wide externalities; extending the framework to a multi-agent or network-based setting is therefore an important avenue for future research. Third, liquidity shocks are modeled as exogenous processes, whereas a richer formulation could endogenize depositor withdrawals by incorporating heterogeneous expectations, strategic complementarities, and game-theoretic run dynamics. Fourth, the model assumes full regulatory compliance and does not explicitly incorporate regulatory arbitrage or temporary constraint breaches that distressed institutions may undertake in order to delay liquidation or preserve short-term funding access. 

Notwithstanding these limitations, the framework provides a theoretically coherent and empirically tractable approach for identifying the boundaries of financial stability and for analyzing the interaction between liquidity stress, credit deterioration, and prudential regulation.

\section*{Acknowledgments}

This work is based upon research funded by Iran National Science Foundation (INSF) under project number No 4049385.

\section*{CRediT authorship contribution statement}

\textbf{Nader Karimi:} Software, Writing -- original draft, Visualization. \textbf{Davood Ahmadian:} Writing -- review \& editing.

\appendix
\section{Derivation of First-Order Conditions...}


\subsection*{Appendix A: Derivation of First-Order Conditions and Kuhn-Tucker}
\label{app:a}

To rigorously prove Lemma~\ref{lem:Threshold}, we define the instantaneous Hamiltonian $\mathcal{H}$. The value function is $V(W, \lambda, d, t)$. The regulatory constraints are $C_1 = LCR_t - 1 \ge 0$ and $C_2 = NSFR_t - 1 \ge 0$.

The Lagrangian is:
\begin{equation}
\mathcal{L} = \mathcal{H} + \mu^{LCR} C_1 + \mu^{NSFR} C_2
\end{equation}
where $\mu^{LCR}, \mu^{NSFR}$ are Kuhn-Tucker multipliers satisfying complementary slackness: $\mu^{LCR} C_1 = 0$ and $\mu^{NSFR} C_2 = 0$.

Taking derivatives w.r.t controls:
\begin{align}
\frac{\partial \mathcal{L}}{\partial b_t} &= r_b W V_W + \mu^{LCR} \frac{\partial LCR}{\partial b_t} + \mu^{NSFR} \frac{\partial NSFR}{\partial b_t} = 0 \label{eq:foc_b_app} \\
\frac{\partial \mathcal{L}}{\partial a_t} &= -\kappa a_t W V_W + \mu^{NSFR} \frac{\partial NSFR}{\partial a_t} + \phi d V_d = 0 \label{eq:foc_a_app}
\end{align}

\begin{proof}[Proof of Lemma~\ref{lem:Threshold}]
Consider the FOC for borrowing $b_t$ in Eq.~\eqref{eq:foc_b_app}. Borrowing ($b_t > 0$) typically increases HQLA (positive for LCR) but decreases Available Stable Funding (negative for NSFR). Thus, $\frac{\partial NSFR}{\partial b_t} < 0$.
When the NSFR constraint is slack, $\mu^{NSFR} = 0$. The FOC for $a_t$ (Eq.~\eqref{eq:foc_a_app}) becomes $-\kappa a_t W V_W + \phi d V_d = 0$. Since $V_W > 0$ and $V_d < 0$ (value decreases with credit risk), and $\phi > 0$, the term $\phi d V_d$ is negative. For the equation to hold with $-\kappa a_t W V_W \le 0$, we must have $a_t = 0$. Thus, the bank uses only borrowing.

When the NSFR constraint binds ($NSFR = 1$), $\mu^{NSFR} > 0$. The term $\mu^{NSFR} \frac{\partial NSFR}{\partial b_t}$ becomes a large negative number (since derivative is negative and multiplier is positive). This forces the borrowing FOC (Eq.~\eqref{eq:foc_b_app}) to be violated if $b_t$ increases, effectively setting $b_t \to 0$. To fill the funding gap, the bank must rely on $a_t$. In the FOC for $a_t$, $\mu^{NSFR} \frac{\partial NSFR}{\partial a_t}$ is positive (since asset sales reduce RSF and thus increase NSFR), allowing $a_t > 0$ to satisfy the equation by offsetting the negative terms. This proves the switching mechanism.
\end{proof}


\begin{thebibliography}{99}

\bibitem{Acharya2015}
Acharya, V. V., \& Mora, N. (2015).
A crisis of banks as liquidity providers.
\textit{The Journal of Finance, 70}(1), 1--43.

\bibitem{AyinuolaGumel2023}
Ayinuola, T. F., \& Gumel, B. I. (2023).
The nexus between liquidity and credit risks and their impact on bank stability.
\textit{Asian Journal of Economics, Business and Accounting, 23}(11), 15--27.

\bibitem{MahisiUsman2024}
Mahisi, P. P. W. N., \& Usman, B. (2024).
The effect of Basel III liquidity, credit risk, and market risk on the profitability of commercial banks in Indonesia.
\textit{Indonesian Interdisciplinary Journal of Sharia Economics (IIJSE), 7}(2), 3416--3440.

\bibitem{Onuoha2024}
Onuoha, D. U. (2024).
Stress testing bank financial systems: A technological perspective.
\textit{International Journal of Scientific Advances, 3}, 339--343.

\bibitem{NathEtAl2024}
Nath, S. D., Biswas, M. R., Maleque, M. A., \& Islam, M. M. (2024).
Effect of Basel III liquidity ratio LCR and NSFR on the profitability of commercial banks in Bangladesh.
\textit{International Journal of Economics \& Business Administration (IJEBA), 12}(3), 12--28.

\bibitem{ClercLecarpentierPouvelle2025}
Clerc, L., Lecarpentier, S., \& Pouvelle, C. (2025).
Basel III joint regulatory constraints: Interactions and implications for the financing of the economy.

\bibitem{SidhuEtAl2022}
Sidhu, A. V., Rastogi, S., Gupte, R., Rawal, A., \& Agarwal, B. (2022).
Net stable funding ratio (NSFR) and bank performance: A study of the Indian banks.
\textit{Journal of Risk and Financial Management, 15}(11), 527.

\bibitem{MontesiPapiro2018}
Montesi, G., \& Papiro, G. (2018).
Bank stress testing: A stochastic simulation framework to assess banks' financial fragility.
\textit{Risks, 6}(3), 82.

\bibitem{Rossi2015}
Rossi, L. (2015).
\textit{Endogenous firms' exit, inefficient banks and business cycle dynamics} (Working Paper No. 99[03-15]).
Universita di Pavia DEM Working Paper Series.

\bibitem{Asuako2025}
Asuako, K. A. (2025).
\textit{Fiscal dominance and bank fragility in frontier economies: Theory and evidence from Sub-Saharan Africa} (Doctoral dissertation, Southern Illinois University at Carbondale).

\bibitem{MaajiEtAl2026}
Ma’aji, M. M., Barnett, C., Bin-Nashwan, S. A., Roslan, N. H., \& Ali, R. A. (2026).
Risk-driven profitability: The role of bank capital, liquidity and credit in frontier banking markets.
\textit{Journal of Financial Regulation and Compliance, 34}(1), 60--80.

\bibitem{Belaid2017}
Belaid, F., Boussaada, R., \& Belguith, H. (2017).
Bank-firm relationship and credit risk: An analysis on Tunisian firms.
\textit{Research in International Business and Finance, 42}, 532--543.

\bibitem{Berrospide2021}
Berrospide, J. M. (2021).
Bank liquidity hoarding and the financial crisis: An empirical evaluation.
\textit{The Quarterly Journal of Finance, 11}(4), 2150020.

\bibitem{Brunnermeier2009}
Brunnermeier, M. K., \& Pedersen, L. H. (2009).
Market liquidity and funding liquidity.
\textit{Review of Financial Studies, 22}(6), 2201--2238.

\bibitem{Christensen2025}
Christensen, J. H. E., \& Steenkamp, D. (2025).
\textit{Joint estimation of liquidity and credit risk premia in bond prices with an application} (Working Paper No. WP/25/01).
South African Reserve Bank Working Paper Series.

\bibitem{Cont2020}
Cont, R., Kotlicki, A., \& Valderrama, L. (2020).
Liquidity at risk: Joint stress testing of solvency and liquidity.
\textit{Journal of Banking \& Finance, 118}, 105871.

\bibitem{Darvas2018}
Darvas, Z., \& Pichler, D. (2018).
\textit{Excess liquidity and bank lending risks in the euro area}.
European Parliament, Policy Department for Economic, Scientific and Quality of Life Policies.

\bibitem{Diamond1983}
Diamond, D. W., \& Dybvig, P. H. (1983).
Bank runs, deposit insurance, and liquidity.
\textit{Journal of Political Economy, 91}(3), 401--419.

\bibitem{Halaj2017}
Hałaj, G., \& Laliotis, D. (2017).
A top-down liquidity stress test framework.
In \textit{STAMP€: Stress-test analytics for macroprudential purposes in the euro area} (pp. 168--191).
European Central Bank.

\bibitem{He2013}
He, Z., \& Krishnamurthy, A. (2013).
Intermediary asset pricing.
\textit{American Economic Review, 103}(2), 732--770.

\bibitem{HosseiniNodeh2022}
Hosseini-Nodeh, Z., Khanjani-Shiraz, R., \& Pardalos, P. M. (2022).
Distributionally robust portfolio optimization with second-order stochastic dominance based on Wasserstein metric.
\textit{Information Sciences, 613}, 828--852.

\bibitem{HosseiniNodeh2023}
Hosseini-Nodeh, Z., Khanjani-Shiraz, R., \& Pardalos, P. M. (2023).
Portfolio optimization using robust mean absolute deviation model: Wasserstein metric approach.
\textit{Finance Research Letters, 54}, 103735.

\bibitem{Imbierowicz2014}
Imbierowicz, B., \& Rauch, C. (2014).
The relationship between liquidity risk and credit risk in banks.
\textit{Journal of Banking \& Finance, 40}, 242--256.

\bibitem{Liao2021}
Liao, Z., Zhang, H., Guo, K., \& Wu, N. (2021).
A network approach to the study of the dynamics of risk spillover in China's bond market.
\textit{Entropy, 23}(7), 920.

\bibitem{Merton1974}
Merton, R. C. (1974).
On the pricing of corporate debt: The risk structure of interest rates.
\textit{The Journal of Finance, 29}(2), 449--470.

\bibitem{Mousa2025}
Mousa, R., Nabil, J., Safty, A., Hassan, I., \& Ibrahim, Y. (2025).
Liquidity--credit risk dynamics and bank profitability: Hybrid econometric and machine learning evidence from MENA.
\textit{Journal of Financial Reporting and Accounting}.

\bibitem{Peykani2025}
Peykani, P., Sargolzaei, M., Tanasescu, C., Shojaie, S. E., \& Kamyabfar, H. (2025).
Investigating the relationship between liquidity risk, credit risk, and solvency risk in banks listed on the Iranian capital market: A panel vector error correction model.
\textit{Economies, 13}(5), 139.

\bibitem{Schularick2020}
Schularick, M., Steffen, S., \& Tröger, T. H. (2020).
\textit{Bank capital and the European recovery from the COVID-19 crisis} (CEPR Discussion Paper No. 14927).
CEPR.

\bibitem{Wang2023}
Wang, Y., Zhang, Y., \& Bashir, U. (2023).
Impact of COVID-19 on the contagion effect of risks in the banking industry: Based on transfer entropy and social network analysis method.
\textit{Risk Management, 25}(2), 1--41.

\bibitem{WojcikMazur2015}
Wójcik-Mazur, A., \& Szajt, M. (2015).
Determinants of liquidity risk in commercial banks in the European Union.
\textit{Argumenta Oeconomica, 2}(35), 25--47.

\end{thebibliography}
\end{document}